\documentclass{sig-alternate-2013}

\usepackage[
  pass,
]{geometry}

\newfont{\mycrnotice}{ptmr8t at 7pt}
\newfont{\myconfname}{ptmri8t at 7pt}
\let\crnotice\mycrnotice%
\let\confname\myconfname%

\usepackage{amsthm,amssymb,amsmath,stmaryrd}
\usepackage{graphics}
\usepackage{bbm}
\usepackage{tikz}
\usepackage[plainpages=false,pdfpagelabels=false,colorlinks=true,citecolor=blue,hypertexnames=false]{hyperref}
\usepackage{color}

\newtheorem{theorem}{Theorem}[section]
\newtheorem{notation}[theorem]{Convention}

\newtheorem{lemma}[theorem]{Lemma}
\newtheorem{remark}[theorem]{Remark}

\newtheorem{definition}[theorem]{Definition}
\newtheorem{example}[theorem]{Example}

\newtheorem{prop}[theorem]{Proposition}

\newcommand{\bF}{ {\mathbb F}}
\newcommand{\bD}{ {\mathbb D}}
\newcommand{\bN}{ {\mathbb N}}
\newcommand{\bW}{ {\mathbb W}}
\newcommand{\bZ}{ {\mathbb Z}}
\newcommand{\cC}{ {\mathcal C}}

\newcommand{\cS}{ {\mathcal S}}

\newcommand{\ta}{\tilde{a}}
\newcommand{\tb}{\tilde{b}}
\newcommand{\tq}{\tilde{q}}

\newcommand{\spanning}{\text{span}}

\def\lc{\operatorname{lc}}
\def\im{\operatorname{im}}

\def\redy{\operatorname{red}_y}

\permission{Permission to make digital or hard copies of all or part of this work for
personal or classroom use is granted without fee provided that copies are not
made or distributed for profit or commercial advantage and that copies bear this
notice and the full citation on the first page. Copyrights for components of
this work owned by others than the author(s) must be honored. Abstracting with
credit is permitted. To copy otherwise, or republish, to post on servers or to
redistribute to lists, requires prior specific permission and/or a fee. Request
permissions from Permissions@acm.org.}
\conferenceinfo{ISSAC'16,}{July 19 - 22 2016, Waterloo, ON, Canada.\\
{\mycrnotice{Copyright is held by the owner/author(s). Publication rights licensed to ACM.}}}
 \copyrightetc{Copyright \copyright~2016 ACM \the\acmcopyr}
\crdata{978-1-4503-4380-0/16/07\ ...\$15.00.\\
DOI: http://dx.doi.org/10.1145/2930889.2930893}

\begin{document}

\title{New Bounds for Hypergeometric Creative Telescoping\titlenote{H.\ Huang was supported by the Austrian Science Fund (FWF) grant W1214-13, three NSFC grants (91118001, 60821002/F02, 11501552) and a 973 project (2011CB302401).}}
%
%
%
%
%

\numberofauthors{1} 
%

\author{\medskip
        Hui Huang$^{1,2}$ \\
        \smallskip
        \affaddr{$^1$KLMM,\, AMSS, \,Chinese Academy of Sciences, Beijing 100190, (China)}\\
        \smallskip
        \affaddr{$^2$Institute for Algebra, Johannes Kepler University, Linz A-4040, (Austria)}\\
        \smallskip
        \email{huanghui@amss.ac.cn}
    }
\maketitle
\begin{abstract}
Based on a modified version of Abramov-Petkov{\v s}ek reduction, a new algorithm to compute minimal telescopers for bivariate hypergeometric terms was developed last year. We investigate further in this paper and present a new argument for the termination of this algorithm, which provides an independent proof of the existence of telescopers and even enables us to derive lower as well as upper bounds for the order of telescopers for hypergeometric terms. Compared to the known bounds in the literature, our bounds are sometimes better, and never worse than the known ones.
\end{abstract}

\category{I.1.2}{Computing Methodologies}{Symbolic and Algebraic Manipulation}[Algebraic
Algorithms]

\terms{Algorithms, Theory}

\keywords{Modified Abramov-Petkov\v{s}ek reduction, Hypergeometric term, Telescoper, Order bound}

\section{Introduction}\label{SEC:intro}
This paper is about creative telescoping for hypergeometric terms.
A~{\em hypergeometric term} is an expression~$f_{x,y}$ in, say, two variables $x,y$ such
that the two shift quotients $f_{x+1,y}/f_{x,y}$ and $f_{x,y+1}/f_{x,y}$ can be
expressed as rational functions in $x$ and~$y$.  The prototypical example of a
hypergeometric term is the binomial coefficient $f_{x,y}=~\binom xy$.  Creative
telescoping is the main tool for simplifying definite sums of hypergeometric
terms. The task consists in finding some nonzero recurrence operator $L$ and another
hypergeometric term~$g_{x,y}$ such that $L\cdot f_{x,y}=g_{x,y+1}-g_{x,y}$. It
is required that the operator $L$ does not contain $y$ or the shift
operator~$\sigma_y$, i.e., it must have the form
$L=e_0+e_1\sigma_x+\cdots+e_\rho\sigma_x^\rho$ for some $e_0,\dots,e_\rho$ that only
depend on~$x$.

If $L$ and $g_{x,y}$ are as above, we say that $L$ is a {\em telescoper} for~$f_{x,y}$, and $g_{x,y}$ is a {\em certificate} for~$L$. Once a telescoper for $f_{x,y}$ is known, we can extract useful information about definite sums such
as $F_x=\sum_{y=0}^x f_{x,y}$ from~$L$. See~\cite{PWZ1996,Zeil1991} for further information. These references also contain classical algorithms for computing telescopers and certificates for given
hypergeometric terms. During the past 25 years, the technique of creative
telescoping has been generalized and refined in various
ways~\cite{MoZe2005,ApZe2006,BCCL2010,BCCLX2013,BLS2013,CHKL2015,CKK2016}. The latest trend in this development are so-called reduction-based algorithms, first presented in~\cite{BCCL2010}.  One of their features
is that they can find a telescoper for a given term $f$ without also computing
the corresponding certificate. This is interesting because a certificate is not
always needed, and it is typically much larger (and thus computationally more
expensive) than the telescoper, so we may not want to compute it if we don't
have to.

Reduction-based algorithms have been first developed in the differential case,
for various cases~\cite{BCCL2010,BCCLX2013,BLS2013,CKK2016}. The basic idea, formulated for the shift
case, is as follows. Let~$\cC$ be a field of characteristic zero. Suppose we know a function~$\redy(\cdot)$, called \emph{reduction}, with the property that for all $f$ in the domain under consideration, say $\bD$, containing $\cC(x,y)$, there exists a $g$ in the same
domain such that $f-\redy(f)=\sigma_y(g)-g$, i.e., the difference $f-\redy(f)$ is a
summable term. We call $\redy(f)$ a {\em remainder} of $f$ with respect to the reduction~$\redy(\cdot)$.
Then in order to find a telescoper for~$f$, we can compute
$\redy(f),\redy(\sigma_x(f)),\redy(\sigma_x^2(f)),\dots$ until we find a linear dependence over
the field $\cC(x)$. If such a dependence is found, say
$e_0\redy(f)+\cdots+e_\rho\redy(\sigma_x^\rho(f))=0$ for some $e_0,\dots,e_\rho$ in $\cC(x)$, then
$e_0+\cdots+e_\rho\sigma_x^\rho$ is a telescoper for~$f$.

In order to show that this method terminates, one possible approach is to show
that the $\cC(x)$-vector space spanned by $\redy(f), \redy(\sigma_x(f)), \redy(\sigma_x^2(f)), \dots$ for $f\in \bD$ has a finite dimension. Then, as soon as $\rho$
exceeds this dimension, we can be sure that a telescoper will be found. This
approach was taken in~\cite{BCCLX2013,BLS2013,CKK2016}.  As a nice side result, this approach
provides an independent proof of the existence of telescopers, and even a bound
on their order.  In the paper from last year~\cite{CHKL2015}, the authors used a different
approach. Instead of showing that the remainders form a finite-dimensional
vector space, they showed that for every summable term $f$, we have $\redy(f)=0$. This also ensures that the method terminates (assuming that we already know
for other reasons that a telescoper exists), and in fact that it will find the
smallest possible telescoper, but it does not provide a bound on its order.

This discrepancy in the approaches for the differential case and the shift case
is unpleasant. It is not clear why the shift case should require a different
argument. The goal of the present paper is to show that it does not. We will
continue the development of last year's theory to a point where we can also show
that the remainders belong to a finite-dimensional vector space. As a result, we
obtain new bounds for the order of telescopers for hypergeometric terms. We
obtain lower as well as upper bounds. We do not find exactly the same bounds
that are already in the literature~\cite{MoZe2005,AbLe2005}. Comparing our bounds to the
known bounds, it appears that for \lq\lq generic\rq\rq\ input, the values often agree (of
course, because the known bounds are already generically sharp). However, there
are some special examples in which our bounds are better than the known
bounds. On the other hand, our bounds are never worse than the old bounds.

\section{Preliminaries}\label{SEC:pre}
Using the same notations as in \cite{CHKL2015}, we let $\bF$ be a field of characteristic zero, and $\bF(y)$ be the field of rational functions in $y$ over $\bF$. Let $\sigma_y$ be the automorphism that maps $r(y)$ to $r(y+1)$ for every $r \in \bF(y)$. The pair $(\bF(y), \sigma_y)$ is called a difference field. A difference ring extension of $(\bF(y),\sigma_y)$ is a ring $\bD$ containing $\bF(y)$ together with a distinguished endomorphism $\sigma_y\colon\bD\to\bD$ whose restriction to $\bF(y)$ agrees with the automorphism defined before. An element $c\in\bD$ is called a constant if $\sigma_y(c)=c$. We denote by $\deg_y(p)$ the degree of a nonzero polynomial $p \in \bF[y]$.
\begin{definition}\label{DEF:ht}
 Let $\bD$ be a difference ring extension of $\bF(y)$. A nonzero element $T \in \bD$ is called
 a {\em hypergeometric term} over $\bF(y)$ if $\sigma_y(T) = r T$ for some $r \in \bF(y)$. We call $r$ the
 {\em shift quotient} of $T$ w.r.t. $y$.
\end{definition}

A univariate hypergeometric term $T$ is called {\em hypergeometric summable} if there exists another hypergeometric term~$G$ s.t. $T = \Delta_y(G)$, where $\Delta_y$ denotes the difference of $\sigma_y$ and the identity map. We abbreviate \lq\lq hypergeometric summable\rq\rq\ as \lq\lq summable\rq\rq\ in this paper.

Recall \cite[\S 1]{AbPe2001a} that a nonzero polynomial in $\bF[y]$ is said to be {\em shift-free} if no two distinct roots differ by an integer. A nonzero rational function in $\bF(y)$ is said to be {\em shift-reduced} if its numerator is co-prime with any shift of its denominator.

According to \cite{AbPe2001a, AbPe2002b}, for a given hypergeometric term $T$ there always exists a rational function $S\in \bF(y)$ and another hypergeometric term~$H$ whose shift quotient is shift-reduced, s.t.\ $T=SH$. This is called a {\em multiplicative decomposition} of~$T$. We call the shift quotient $K = \sigma_y(H)/H$ a {\em kernel} of~$T$ and $S$ the corresponding {\em shell}.

Based on Abramov and Petkov{\v s}ek's work in~\cite{AbPe2001a,AbPe2002b}, the authors of~\cite{CHKL2015} presented a modified version of Abramov-Petkov{\v s}ek reduction, which determines summability without solving any auxiliary difference equations. To describe it concisely, we first recall some terminology.

Let~$T$ be a hypergeometric term whose kernel is~$K$ and the corresponding shell is~$S$. Then~$T=SH$, where~$H$ is a hypergeometric term whose shift quotient is~$K$. Write~$K = u/v$, where~$u,v$ are polynomials in~$\bF[y]$ with~$\gcd(u,v)=1$.

\begin{definition}\label{DEF:prime}
 A nonzero polynomial $p$ in $\bF[y]$ is said to be {\em strongly prime} with $K$
 if $\gcd\left( p, \sigma_y^{-i}(u) \right) {=} \gcd\left( p, \sigma_y^{i}(v) \right) {=} 1$
 for all $i \ge 0$.
\end{definition}

Now define the $\bF$-linear map $\phi_K$ from $\bF[y]$ to itself by sending $p$ to $u\sigma_y(p)-v p$ for all $p\in \bF[y]$. We call $\phi_K$ the {\em map for polynomial reduction w.r.t.\ $K$}. Let
\[\bW_K = \spanning\{y^\ell \mid \ell \in \bN, \ell \neq \deg_y(p) \text{ for all } p \in \im(\phi_K)\}.\]
Then $\bF[y] = \im(\phi_K) \oplus \bW_K$, and thus we call $\bW_K$ the {\em standard complement of $\im(\phi_K)$}.
    \begin{definition} \label{DEF:residual}
        Let~$f$ be a rational function in $\bF(y)$. Another rational function~$r$ in~$\bF(y)$ is called a {\em (discrete) residual form} of~$f$ w.r.t.~$K$ if there exists~$g\in \bF(y)$ and~$a,b,q$ in~$\bF[y]$ s.t.
        $$f = K \sigma_y(g)-g + r \quad \text{and} \quad r = \frac{a}{b}+\frac{q}{v},$$
        where~$\deg_y(a)<\deg_y(b)$, $\gcd(a,b)=1$, $b$ is shift-free and strongly
        prime with~$K$, and~$q$ belongs to~$\bW_K$. For brevity, we just say that~$r$ is {\em a residual form} w.r.t.~$K$ if~$f$ is clear from the context. We call~$b$ the {\em significant denominator} of~$r$.
    \end{definition}
The modified Abramov-Petkov{\v s}ek reduction \cite[Theorem~4.8]{CHKL2015} can be stated as follows.
\begin{theorem}\label{THM:map}
 With the notations given above, the modified version of the Abramov-Petkov{\v s}ek reduction computes a rational function $g$ in $\bF(y)$ and a residual form $r$ w.r.t.\ $K$, such that
 \begin{equation} \label{EQ:map}
   T = \Delta_y(g H) +  r H.
 \end{equation}
 Moreover, $T$ is summable if and only if $r=0$.
\end{theorem}

\section{Properties of residual forms}\label{SEC:remprops}
In this section, we will explore important properties of residual forms, which enables us to derive nontrivial relationship among remainders in Section~\ref{SEC:remainder}.

Unlike the differential case, a rational function may have more than one residual form in shift case. However, these residual forms are related to each other in some way. Before describing it, let us recall some technology.

Recall~\cite[\S 2]{AbPe2002b} that polynomials~$p_1, p_2 \in \bF[y]$ are said to be {\em shift-equivalent} (w.r.t.\ $y$) if~$p_1 = \sigma_y^\ell(p_2)$ for some~$\ell \in \bZ$, denoted by~$p_1 \sim_y p_2$. It is an equivalence relation.

Let~$f$ be a rational function in $\bF(y)$. We call the rational function pair~$(K, S) \in \bF(y)^2$ a {\em rational normal form (RNF)} of~$f$ if~$f = K \cdot \sigma_y(S)/S$ and~$K$ is shift-reduced. By \cite[Theorem~1]{AbPe2002b}, every rational function has at least one RNF. Let~$T$ be a hypergeometric term over~$\bF(y)$. It is not hard to see that $(K,S)\in \bF(y)^2$ is an RNF of~$\sigma_y(T)/T$ if and only if~$K$ and $S$ are a kernel and the corresponding shell of~$T$.

\begin{definition}\label{DEF:relatedness}
	Two shift-free polynomials $p, q \in \bF[y]$ are called {\em shift-related} (w.r.t.\ $y$), denoted by $p \approx_y q$, if for any nontrivial monic irreducible factor $f$ of $p$, there exists a unique monic irreducible factor $g$ of $q$ with the same multiplicity as $f$ in $p$ s.t.\ $f \sim_y g$, and vice versa.
\end{definition}
One can show that~$\approx_y$ is an equivalence relation.
\begin{prop}\label{PRO:uniqueness}
	Let~$K$ be a shift-reduced rational function in~$\bF(y)$. Assume that~$r_1, r_2\in \bF(y)$ are both residual forms of the same rational function in~$\bF(y)$ w.r.t.\ $K$. Then the significant denominators of~$r_1$ and~$r_2$ are shift-related to each other.
\end{prop}
\begin{proof}
Assume that~$r_1,r_2$ are of the forms
\begin{equation*}\label{EQ:drfs}
r_1 = \frac{a_1}{b_1} + \frac{q_1}{v} \quad \text{and}\quad r_2 = \frac{a_2}{b_2} + \frac{q_2}{v},
\end{equation*}
where~$a_i, b_i \in \bF[y]$, $\deg(a_i) < \deg(b_i)$,~$\gcd(a_i,b_i) = 1$,~$b_i$ is shift-free and strongly prime with~$K$, $q_i\in \bW_{K}$ for $i = 1,2$, and~$v$ is the denominator of~$K$. Since~$r_1, r_2$ are both residual forms of the same rational function, there exists $g\in\bF(y)$ s.t.
\[r_2 = K \sigma_y(g) - g + r_1.\]
It follows that
\begin{equation}\label{EQ:drfrel}
\frac{a_2 v}{b_2} = u \sigma_y(g) - v (g) + (q_1 - q_2) + \frac{a_1 v}{b_1}.
\end{equation}
Let~$f\in \bF[y]$ be a nontrivial monic irreducible factor of~$b_1$ with multiplicity~$\alpha>0$. If~$f^\alpha$ divides~$b_2$, then we are done. Otherwise, let~$g_d$ be the denominator of~$g$. Then~$f^\alpha$ divides~$g_d$ or~$\sigma_y(g_d)$ as~$\gcd(b_1, a_1 v) = 1$. If $f^\alpha$ divides~$g_d$, let
\begin{equation*}
 m =  \max \{\ell \in \bZ \mid \sigma_y^{\ell}(f)^\alpha \text{ divides } g_d\}.
\end{equation*}
Then~$m \geq 0$ and~$\sigma_y^{m+1}(f)^\alpha \mid \sigma_y(g_d)$. Since~$b_1$ is strongly prime with~$K$,~$\gcd(\sigma_y^{m+1}(f)^\alpha, u) = 1$. Apparently, neither~$b_1$ nor~$g_d$ is divisible by~$\sigma_y^{m+1}(f)^\alpha$ as~$b_1$ is shift-free and~$m$ is maximal. Hence~\eqref{EQ:drfrel} implies~$\sigma_y^{m+1}(f)^\alpha$ is the required factor of~$b_2$. Similarly, we can show that~$\sigma_y^{m}(f)^\alpha$ with
\[m =  \min \{\ell \in \bZ \mid \sigma_y^{\ell}(f)^\alpha \text{ divides } g_d\} \leq -1,\]
is the required factor of~$b_2$, if~$f^\alpha$ divides~$\sigma_y(g_d)$.

In summary, there always exists a monic irreducible factor of~$b_2$ with multiplicity at least~$\alpha$ s.t.\ it is shift-equivalent to~$f$. Due to the shift-freeness of~$b_2$, this factor is unique. Conversely, the proof proceeds in a similar way as above. According to the definition,~$b_1\approx_y b_2$.
\end{proof}
Given a hypergeometric term, it is readily seen that the above proposition reveals the relationship between two residual forms w.r.t.\ the same kernel. To study the case with different kernels, we need to develop two lemmas.

\begin{lemma}\label{LEM:vrelatedness}
    Let~$(K,S)$ be an RNF of a rational function~$f$ in~$\bF(y)$ and~$r$ be a residual form of~$S$ w.r.t.\ $K$. Write
    \[K = \frac{u}{v} \quad \text{with}\ u, v\in \bF[y]\text{ and } \gcd(u,v)=1.\]
    Assume that $p\in \bF[y]$ is a nontrivial monic irreducible factor of~$v$ with multiplicity~$\alpha> 0$. Then
    \[(K', S') = \left(\frac{u}{v'\sigma_y(p)^\alpha}, p^\alpha S\right)\]
    is an RNF of~$f$, in which~$v'= v/p^\alpha$. Moreover, there exists a residual form~$r'$ of~$S'$  w.r.t.\ $K'$ whose significant denominator is equal to that of~$r$.
\end{lemma}
\begin{proof}
Since~$K$ is shift-reduced, so is~$K'$. Then the first assertion follows by noticing
 \[K\frac{\sigma_y(S)} {S} = \frac{u}{v' p^\alpha}\frac{\sigma_y(S)}{S} = \frac{u}{v' \sigma_y(p)^\alpha}\frac{\sigma_y(p^\alpha S) }{p^\alpha S}=K'\frac{\sigma_y(S')}{S'}.\]
Let~$r$ be of the form~$r = a/b+q/v$ where~$a, b \in \bF[y]$, $\deg(a) < \deg(b)$,~$\gcd(a,b) = 1$,~$b$ is shift-free and strongly prime with~$K$, and $q\in \bW_{K}$. Then there exists~$g\in \bF(y)$ s.t.
\[S = K \sigma_y(g) - g + \frac{a}{b} + \frac{q}{v'p^\alpha},\]
which implies that
\begin{align*}
      S' &= p^\alpha S = p^\alpha K \sigma_y(g) - p^\alpha g + \frac{ap^\alpha }{b} + \frac{q}{v'}\\
      &= \frac{u}{v'\sigma_y(p)^\alpha} \sigma_y(p^\alpha g) - p^\alpha g + \frac{ap^\alpha}{b} + \frac{q\sigma_y(p)^\alpha}{v'\sigma_y(p)^\alpha}\\
       &=K'\sigma_y(p^\alpha g) - p^\alpha g + \frac{ap^\alpha }{b} + \frac{q\sigma_y(p)^\alpha}{v'\sigma_y(p)^\alpha}
\end{align*}
Since~$b$ is strongly prime with~$K$ and~$\gcd(a,b)=1$, we have~$\gcd(ap^\alpha, b)=1$. According to Lemma~4.2 and Remark~4.3 in~\cite{CHKL2015}, there exist~$g'\in \bF(y), a', q' \in \bF[y]$ with $\deg_y(a') < \deg_y(b)$ and $\gcd(a',b)=1$, and~$q' \in \bW_{K'}$ s.t.
\[S'= K'\sigma_y(g')-g' + \left(\frac{a'}{b}+ \frac{q'}{v' \sigma_y(p)^\alpha}\right).\]
Note that~$b$ is strongly prime with~$K$, so~$b$ is also strongly prime with~$K'$. By the shift-freeness of~$b$,
\[\frac{a'}{b}+ \frac{q'}{v' \sigma_y(p)^\alpha}\]
is a residual form of~$S'$ w.r.t.\ $K'$. The lemma follows.
\end{proof}
\begin{lemma}\label{LEM:urelatedness}
    Let~$(K,S)$ be an RNF of a rational function~$f$ in~$\bF(y)$ and~$r$ be a residual form of~$S$ w.r.t.\ $K$. Write
    \[K = \frac{u}{v} \quad \text{with}\ u, v\in \bF[y]\text{ and } \gcd(u,v)=1.\]
    Assume that $p\in \bF[y]$ is a nontrivial monic irreducible factor of~$u$ with multiplicity~$\alpha> 0$. Then
    \[(K', S') = \left(\frac{u'\sigma_y^{-1}(p)^\alpha }{v}, \sigma_y^{-1}(p)^\alpha S\right)\]
    is an RNF of~$f$, in which~$u'= u/p^\alpha$. Moreover, there exists a residual form~$r'$ of~$S'$  w.r.t.\ $K'$ whose significant denominator is equal to that of~$r$.
\end{lemma}
\begin{proof}
    Similar to Lemma~\ref{LEM:vrelatedness}.
\end{proof}
\begin{prop} \label{PROP:sd}
    Let~$(K,S)$ be an RNF of a rational function~$f$ in~$\bF(y)$ and $r$ be a residual form of~$S$ w.r.t.\ $K$. Then there exists another RNF $(\tilde K, \tilde S)$ of~$f$ such that
\begin{enumerate}
\item $\tilde K$ has shift-free numerator and shift-free denominator;
\item there exists a residual form~$\tilde r$ of~$\tilde S$ w.r.t.\ $\tilde K$ whose significant denominator is equal to that of~$r$.
\end{enumerate}
\end{prop}
\begin{proof}
Let~$K=u/v$ with~$u, v\in \bF[y]$ and~$\gcd(u,v)=1$, and~$b$ be the significant denominator of~$r$.

Assume that~$v$ is not shift-free. Then there exist two nontrivial monic irreducible factors~$p$ and~$\sigma_y^m(p)$ $(m>0)$ of~$v$ with multiplicity~$\alpha>0$ and~$\beta>0$, respectively. W.L.O.G.,
suppose further that~$\sigma_y^{\ell}(p)$ is not a factor of~$v$ for all~$\ell <0$ and~$\ell >m$. By Lemma~\ref{LEM:vrelatedness}, $f$ has an RNF~$(K^\prime, S^\prime)$, in which~$K^\prime$ has a denominator~$v^\prime=\tilde{v} \sigma_y(p)^\alpha$,
where~$\tilde{v}=v/p^\alpha$, and the numerator remains to be~$u$. Moreover, there exists a residual form of~$S^\prime$ w.r.t.\ $K^\prime$ whose significant denominator is~$b$.
If~$m=1$, $\sigma_y(p)$ is an irreducible factor of~$v^\prime$ with multiplicity~$\alpha+\beta$. Otherwise, it is an irreducible factor of~$v^\prime$ with multiplicity~$\alpha$.
More importantly,~$\sigma_y^{\ell}(p)$ is not a factor of~$v^\prime$ for all~$\ell <1$. Iteratively using the argument, we arrive at an RNF of~$f$ such that
$\sigma_y^m(p)$ divides the denominator of the new kernel with certain multiplicity but~$\sigma_y^i(p)$ does not whenever~$i\neq m$.  Moreover, there exists a residual form of the new shell with
respect to the new kernel whose significant denominator is equal to~$b$. Applying the same argument to each irreducible factor, we can obtain an RNF of~$f$ whose kernel has a shift-free
denominator and whose shell has a residual form with significant denominator~$b$.

With Lemma~\ref{LEM:urelatedness}, one can obtain an RNF of~$f$ whose kernel has a shift free numerator whose shell has a residual form with significant denominator~$b$.
\end{proof}
A nonzero rational function is said to be {\em shift-free} if it is shift-reduced and its denominator and numerator are both shift-free. The main result is given below.
\begin{prop} \label{PRO:sf}
Let~$(K,S)$ and~$(K^\prime, S^\prime)$ be two RNF's of a rational function~$f$ in~$\bF(y)$, $r$ and~$r^\prime$ be residual forms of~$S$ (w.r.t.\ $K$) and~$S^\prime$ (w.r.t.\ $K'$), respectively.
Then the significant denominators of~$r$ and~$r^\prime$ are shift-related.
\end{prop}
\begin{proof}
Let~$b$ and~$b^\prime$ be the significant denominators of~$r$ and~$r^\prime$, respectively.
By the above proposition, there exist two RNF's~$(\tilde K, \tilde S)$ and~$(\tilde{K}^\prime, \tilde{S}^\prime)$ of~$f$ such that their kernels are shift-free and
their shells have residual forms whose significant denominators are~$b$ and~$b^\prime$, respectively.

According to~\cite[Theorem~2]{AbPe2002b}, the respective denominators~$\tilde v$ and~$\tilde{v}^\prime$ of~$\tilde K$ and~$\tilde{K}^\prime$ are shift-related. Thus, for a nontrivial monic irreducible factor~$p$ of~$\tilde v$ with multiplicity~$\alpha>0$, there exists a unique factor~$\sigma_y^\ell(p)$ of~$\tilde{v}^\prime$ with the same multiplicity. W.L.O.G., we may assume~$\ell \le 0$. Otherwise, we can switch the roles of~$(\tilde K, \tilde S)$ and~$(\tilde{K}^\prime, \tilde{S}^\prime)$.
If~$\ell<0$, a repeated use of Lemma~\ref{LEM:vrelatedness} leads to a new RNF~$(\tilde{K}^{\prime\prime}, \tilde{S}^{\prime\prime})$ from~$(\tilde{K}^\prime, \tilde{S}^\prime)$
such that~$\tilde{K}^{\prime\prime}$ is shift-free, $p$ is a factor of the denominator of~$\tilde{K}^{\prime\prime}$ with the same multiplicity.

Applying the above argument to each irreducible factor and using Lemma~\ref{LEM:urelatedness} for numerators in the same fashion, we can obtain two new RNF's whose kernels are equal and whose shells
have respective residual forms with significant denominators~$b$ and~$b^\prime$. It follows that~$b$ and~$b^\prime$ are shift-related by Proposition~\ref{PRO:uniqueness}.
\end{proof}

\section{Telescoping via reductions}\label{SEC:telescoping}
We now translate terminology concerning univariate hypergeometric terms to bivariate ones. Let $\cC$ be a field of characteristic zero, and $\cC(x,y)$ be the field of rational functions in $x$ and $y$ over $\cC$. Let $\sigma_x, \sigma_y$ be the shift operators w.r.t.\ $x$ and $y$, respectively, defined by,
\[\sigma_x(f(x,y))=f(x+1,y) \text{ and } \sigma_y(f(x,y)) = f(x,y+1),\]
for any $f$ in $\cC(x,y)$. Then the pair $(\cC(x,y), \{\sigma_x,\sigma_y\})$ forms a partial difference field.
\begin{definition}\label{DEF:biht}
 Let $\bD$ be a partial difference ring extension of~$\cC(x,y)$. A nonzero element $T\in \bD$ is called a {\em hypergeometric term} over $\cC(x,y)$ if there exist $f,g\in \cC(x,y)$ s.t.~$\sigma_x(T) = fT$ and $\sigma_y(T) = gT$. We call $f$ and $g$ the $x$-shift quotient and $y$-shift quotient of $T$, respectively.
\end{definition}
An irreducible polynomial $p$ in $\cC[x,y]$ is called {\em integer-linear} over $\cC$ if there exists a univariate polynomial $P\in \cC[z]$ and two integers $\lambda,\mu$ s.t.\ $p = P(\lambda x+\mu y)$. A polynomial in~$\cC[x,y]$ is called {\em integer-linear} over $\cC$ if all its irreducible factors are integer-linear. A rational function in $\cC(x,y)$ is called {\em integer-linear} over $\cC$ if its denominator and numerator are both integer-linear.

Let $\cC(x)\langle S_x\rangle$ be the ring of linear recurrence operators in~$x$, in which the commutation rule is that $S_x r = \sigma_x(r) S_x$ for all~$r \in \cC(x)$. The application of an operator $L = \sum_{i=0}^\rho e_i S_x^i\in \cC(x)\langle S_x\rangle$ to a hypergeometric term $T$ is defined as $L(T) = \sum_{i=0}^\rho e_i\sigma_x^i(T)$.

Given a hypergeometric term $T$ over $\cC(x,y)$, the computational problem of creative telescoping is to construct a nonzero operator $L\in \cC(x)\langle S_x\rangle$ s.t.
\[L(T) = \Delta_y(G),\]
for some hypergeometric term $G$. We call $L$ a {\em telescoper} for~$T$ w.r.t.\ $y$ and $G$ a {\em certificate} for~$L$. To avoid unnecessary duplication, we make a convention.
\begin{notation}\label{CON:convention}
  Let~$T$ be a hypergeometric term over $\cC(x, y)$ with a multiplicative decomposition $SH$, where $S$ is in $\cC(x,y)$ and $H$ is a hypergeometric term whose $y$-shift quotient~$K$ is shift-reduced w.r.t.~$y$. 
  By~\cite[Theorem~8]{AbPe2001b}, we know~$K$ is integer-linear over~$\cC$. Write~$K = u/v$ where~$u, v\in \cC(x)[y]$ and~$\gcd(u,v)=1$.
\end{notation}
For hypergeometric terms, telescopers do not always exist. Abramov presented a criterion for determining the existence of telescopers in~\cite[Theorem 10]{Abra2003}. With Convention~\ref{CON:convention}, applying the modified Abramov-Petkov{\v s}ek reduction to $T$ w.r.t.~$y$ yields \eqref{EQ:map}. By Abramov's criterion, $T$ has a telescoper if and only if the significant denominator of~$r$ in \eqref{EQ:map} is integer-linear over $\cC$. Based on this criterion and the modified reduction, the authors of \cite{CHKL2015} developed a reduction-based telescoping algorithm, named~{\sf ReductionCT}, which either finds a minimal telescoper for $T$, or proves that no telescoper exists. The key advantage of this algorithm is that it separates the computation of telescopers from that of certificates. This is desirable in the typical situation where we are only interested in the telescopers and their size is much smaller than that of certificates.

When the existence of telescopers for $T$ is guaranteed,
we summarize below the idea of the algorithm {\sf ReductionCT}.

We begin by fixing the order of a telescoper for $T$, say~$\rho$, and then look for a telescoper of that order. If none exists, we look for one of the next higher order. We make an ansatz
\[L = e_0 + e_1 S_x + \dots + e_\rho S_x^\rho\]
with undetermined coefficients $e_0, \dots, e_\rho \in \cC(x)$. For $i$ from~$0$ to $\rho$, iteratively applying the modified reduction to~$\sigma_x^i(T)$ and manipulating the resulting residual forms according to Theorem~5.6 in~\cite{CHKL2015} lead to
\begin{equation}\label{EQ:ithdecom}
 \sigma_x^i(T) = \Delta_y(g_i H) +  \left(\frac{a_i}{b_i} + \frac{q_i}{v}\right)H,
\end{equation}
where $g_i\in \cC(x,y)$, $a_i, b_i\in \cC(x)[y]$ with $\deg_y(a_i)<\deg_y(b_i)$, $\gcd(a_i,b_i)=1$, $b_i$ is shift-free w.r.t.\ $y$ and strongly prime with $K$, and $q_i$ belongs to $\bW_K$. Moreover, the least common multiple~$B_\rho$ of~$b_0, \dots, b_\rho$ is shift-free w.r.t.\ $y$.
Let
\[A_\rho = \sum_{i=0}^\rho e_i a_i\frac{B_\rho}{b_i}\quad \text{and}\quad Q_\rho = \sum_{i=0}^\rho e_iq_i.\]
Then~$\deg_y(A_\rho) < \deg_y(B_\rho)$, $B_\rho$ is shift-free w.r.t.~$y$ and strongly prime with $K$. Moreover,~$\bW_K$ is a linear space over~$\cC(x)$, so $Q$ is in~$ \bW_K$. A direct calculation shows that
\[L(T) = \Delta_y\left(\sum_{i=0}^\rho e_i g_i H\right) +  \left(\frac{A_\rho}{B_\rho} + \frac{Q_\rho}{v}\right)H.\]
According to Theorem~\ref{THM:map}, $L(T)$ is summable w.r.t.\ $y$ if and only if $A_\rho/B_\rho+ Q_\rho/v= 0$.
Equivalently, $L$ is a telescoper for $T$ if and only if the linear system
\begin{equation}\label{EQ:system}
\begin{cases}
A_\rho =e_0 a_0\frac{B_\rho}{b_0} \ +\ e_1 a_1\frac{B_\rho}{b_1}\ +\ \cdots\ +\ e_\rho a_\rho\frac{B_\rho}{b_\rho}=0\\[1ex]
Q_\rho =e_0q_0\ +\ e_1q_1\ +\ \cdots\ +\ e_\rho q_\rho=0
\end{cases}
\end{equation}
has a nontrivial solution in $\cC(x)^{\rho+1}$. A linear dependence among these residual forms $\{a_i/b_i+q_i/v\}_{i= 0}^\rho$, for minimal~$\rho$, gives rise to a minimal telescoper for $T$.

The termination of the algorithm {\sf ReductionCT} is guaranteed by Abramov's criterion, see Theorem~6.3 in~\cite{CHKL2015} for more details. However, instead of using Abramov's criterion, one could prove the algorithm {\sf ReductionCT} terminates by showing that the residual forms $\{a_i/b_i+q_i/v\}_{i\geq 0}$ from~\eqref{EQ:ithdecom} form a finite-dimensional vector space over $\cC(x)$. This is exactly what we are going to do in the next section.

\section{Finite-dimensional remainders}\label{SEC:remainder}
In this section, we will show that some sequence of~$\{b_i\}_{i\geq 0}$ satisfying~\eqref{EQ:ithdecom} has a common multiple $B$, provided that~$T$  has  a telescoper. Moreover, $B$ is shift-free and strongly prime with $K$. The existence of this common multiple implies that the corresponding $\{a_i/b_i + q_i/v\}_{i\geq 0}$ from~\eqref{EQ:ithdecom} span a finite-dimensional vector space over $\cC(x)$, and lead to upper and lower bounds on the order of minimal telescopers. To this end, we need some preparations.
\subsection{Shift-homogeneous decomposition}\label{SUBSEC:shifthomdecomp}
Recall \cite{AbPe2001a} that irreducible polynomials $p, q$ in~$\cC[x,y]$ are said to be {\em shift-equivalent} w.r.t.\ $x,y$, denoted by $p \sim_{x,y} q$, if there exist two integers $m, n$ such that $q=\sigma_x^m\sigma_y^n(p)$.
Clearly $\sim_{x,y}$ is an equivalence relation. Choosing the pure lexicographic order~$x\prec y$, we say a polynomial is {\em monic} if its highest term has coefficient~$1$. A rational function is said to be {\em shift-homogeneous} if all non-constant monic irreducible factors of its denominator and numerator belong to the same shift-equivalence class.

By grouping together the factors in the same shift-equivalence class, every rational function $r\in \cC(x,y)$ can be decomposed into the form
\begin{equation}\label{EQ:shifthomdecomp}
r = c\, r_1 \dots r_s
\end{equation}
where $c \in \cC$, $s \in \bN$, each $r_i$ is a shift-homogeneous rational function, and any two non-constant monic irreducible factors of $r_i$ and $r_j$ are pairwise shift-inequivalent whenever $i \neq j$. We call \eqref{EQ:shifthomdecomp} a {\em shift-homogeneous decomposition} of $r$. The shift-homogeneous decomposition is unique up to the order of the factors and multiplication by nonzero constants.

Let $p\in \cC[x,y]$ be an irreducible integer-linear polynomial. Then $p=P(\lambda x + \mu y)$ for~$P \in \cC[z]$ and $\lambda, \mu\in \bZ$. W.L.O.G., we further assume that~$\mu \geq 0$ and~$\gcd(\lambda, \mu ) =1$. By B{\' e}zout's relation, there exist unique integers $\alpha, \beta$ with $|\alpha| <|\mu|$ and~$|\beta| <|\lambda|$ such that $\alpha \lambda + \beta \mu = 1$. Define~$\delta^{(\lambda, \mu)}$ to be $\sigma_x^\alpha \sigma_y^\beta$. For brevity, we just write $\delta$ if $(\lambda, \mu)$ is clear from the context. Note that $\delta(P(z))=P(z+1)$ with $z=\lambda x + \mu y$, which allows us to treat integer-linear polynomials as univariate ones.  For a Laurent polynomial~$\xi=\sum_{i=\ell}^\rho m_i \delta^i$ in~$\bZ[\delta,\delta^{-1}]$ with~$\ell, \rho, m_i\in \bZ$ and $\ell \leq\rho$, define
\[p^\xi = \delta^{\ell}(p^{m_{\ell}}) \delta^{\ell+1}(p^{m_{\ell+1}})  \cdots \delta^\rho(p^{m_\rho}).\]
It is readily seen that for any two irreducible integer-linear polynomials~$p, q\in \cC[x,y]$ of the forms $p=P(\lambda_1 x+\mu_1 y)$ and~$q=Q(\lambda_2 x+ \mu_2 y)$ with~$P, Q\in \cC[z]$, $\lambda_1, \mu_1, \lambda_2, \mu_2 {\in} \bZ$, $\mu_1, \mu_2{\geq} 0$ and $\gcd(\lambda_1, \mu_1) {=} \gcd(\lambda_2, \mu_2) {=}1$, we have~$p \sim_{x,y} q$ if and only if $\lambda_1 = \lambda_2$, $\mu_1=\mu_2$ and $q = p^{\delta^k}$ for some integer~$k$, in which $\delta = \delta^{(\lambda_1, \mu_1)} =\delta^{(\lambda_2, \mu_2)}$.

Adapt from~\eqref{EQ:shifthomdecomp}, every integer-linear rational function~$r$ in~$\cC(x,y)$ admits the following decomposition
\begin{equation}\label{EQ:integerlineardecomp}
r = c_r\, h_1^{\xi_1} \cdots h_s^{\xi_s}
\end{equation}
where $c_r\in \cC$, $s\in \bN$, each~$h_i\in \cC[x,y]$ is irreducible, monic and integer-linear over~$\cC$, and then $h_i =P_i(\lambda_i x + \mu_i y)$ for~$P_i \in \cC[z]$, $\lambda_i, \mu_i\in \bZ$ with $\mu_i\geq 0$, $\gcd(\lambda_i, \mu_i)=1$, and~$\xi_i\in \bZ[\delta^{(\lambda_i, \mu_i)},(\delta^{(\lambda_i, \mu_i)})^{-1}]$. Moreover, $h_i \nsim_{x,y} h_j $ whenever~$i \neq j$. W.L.O.G., we further assume that~$\xi_i$ belongs to $\bZ[\delta^{(\lambda_i, \mu_i)}]$.
\subsection{Relationship among remainders}\label{SUBSEC:relationship}
With Proposition~\ref{PRO:sf}, we can describe an inherent relationship among any residual forms $\{a_i/b_i+q_i/v\}_{i\geq 0}$ satisfying~\eqref{EQ:ithdecom}.
\begin{lemma}\label{LEM:xrfs}
    With Convention~\ref{CON:convention}, let~$r$ be a residual form of~$S$ w.r.t.\ $K$. Then~$\sigma_x(K)$ and $\sigma_x(S)$ are a kernel and the corresponding shell of~$\sigma_x(T)$ w.r.t.\ $y$. Moreover, $\sigma_x(r)$ is a residual form of~$\sigma_x(S)$ w.r.t.\ $\sigma_x(K)$.
\end{lemma}
\begin{proof}
      According to Convention~\ref{CON:convention}, $\sigma_x(T) = \sigma_x(S)\sigma_x(H)$ and $\sigma_x(K)$ is the $y$-shift quotient of~$\sigma_x(H)$. To prove the first assertion, one needs to show that~$\sigma_x(K)$ is shift-reduced w.r.t.~$y$, which can be proven by observing that, for any two polynomials~$p_1, p_2 \in \cC(x)[y]$, $\gcd(\sigma_x(p_1),\sigma_x(p_2)) = 1$ if and only if $\gcd(p_1,p_2) = 1$. Let~$r = a/b+q/v$, where~$a,b,q$ belong to $\cC(x)[y]$, $\deg_y(a)<\deg_y(b)$, $\gcd(a,b)=1$, $b$ is shift-free and strongly prime with~$K$, and~$q\in \bW_K$. Clearly, we have $\deg_y(\sigma_x(a))<\deg_y(\sigma_x(b))$ and $\gcd(\sigma_x(a),\sigma_x(b))=1$. The shift-freeness and strong primeness w.r.t.\ $\sigma_x(K)$ of~$\sigma_x(b)$ easily follows by the above observation.

      Note that~$\sigma_x\circ \deg_y = \deg_y \circ \,\sigma_x$ and $\sigma_x\circ \lc_y = \lc_y\circ \,\sigma_x$, where~$\lc_y(p)$ is the leading coefficient of a polynomial~$p$ in~$\cC(x)[y]$. So the standard complements~$\bW_K$ and $\bW_{\sigma_x(K)}$ for polynomial reduction have the same echelon basis according to the case study in~\cite[\S 4.2]{CHKL2015}. It follows from~$q\in \bW_K$ that $\sigma_x(q) \in \bW_{\sigma_x(K)}$. Accordingly, $\sigma_x(r)$ is a residual form of~$\sigma_x(S)$ w.r.t.\ $\sigma_x(K)$.
\end{proof}
\begin{prop}\label{PRO:brelationship}
    With Convention \ref{CON:convention}, for every nonnegative integer $i$, assume that $\sigma_x^i(T)$ can be decomposed into~\eqref{EQ:ithdecom},
    where $g_i\in \cC(x,y)$, $a_i, b_i \in \cC(x)[y]$ with~$\deg_y(a_i)<\deg_y(b_i)$, $\gcd(a_i,b_i)=1$, $b_i$ shift-free w.r.t.\ $y$ and strongly prime with~$K$, and~$q_i$ belongs to $\bW_K$. Then~$b_i\approx_y \sigma_x^i(b_0)$.
\end{prop}
\begin{proof}
   It suffices to show $b_1\approx_y \sigma_x(b_0)$. The rest follows by a direct induction on $i$.

   Applying $\sigma_x$ to both sides of \eqref{EQ:ithdecom} with $i = 0$ gives
   \begin{align*}
     \sigma_x(T) &= \sigma_x(\Delta_y(g_0 H))+\sigma_x\left(\frac{a_0}{b_0} + \frac{q_0}{v}\right)\sigma_x(H)\\
     &=\Delta_y(\sigma_x(g_0 H))+\left(\frac{\sigma_x(a_0)}{\sigma_x(b_0)} + \frac{\sigma_x(q_0)}{\sigma_x(v)} \right)\sigma_x(H)
   \end{align*}
   It follows from Lemma~\ref{LEM:xrfs} that $(\sigma_x(K), \sigma_x(S))$ is an RNF of the $y$-shift quotient of~$\sigma_x(T)$, and $\sigma_x(a_0)/\sigma_x(b_0) + \sigma_x(q_0)/\sigma_x(v)$ is a residual form of~$\sigma_x(S)$ w.r.t.\ $\sigma_x(K)$. Let~$N = \sigma_x(H)/H$. Then $(K,\sigma_x(S)N)$ is also an RNF of the $y$-shift quotient of~$\sigma_x(T)$. By~\eqref{EQ:ithdecom} with~$i = 1$, $a_1/b_1 + q_1/v$ is a residual form of~$\sigma_x(S)N $ w.r.t.~$K$. By Proposition~\ref{PRO:sf}, we have $\sigma_x(b_0)\approx_y b_1$.
\end{proof}
The following lemma says that, with Convention~\ref{CON:convention}, for any polynomial~$f$ in $\cC(x)[y]$, there always exists~$g \in \cC(x)[y]$ s.t.\ $f\approx_y g$ and $g$ is strongly prime with~$K$.
\begin{lemma}\label{LEM:stronglyprime}
    With Convention~\ref{CON:convention}, assume that~$p$ is an irreducible polynomial in~$\cC(x)[y]$. Then there exists an integer~$m$ s.t.\ $\sigma_y^m(p)$ is strongly prime with~$K$.
\end{lemma}
\begin{proof}
    It suffices to consider the following three cases according to the definition of strong primeness.

    \smallskip\noindent
    {\em Case~1.} $p$ is strongly prime with~$K$. Then the lemma follows by letting~$m = 0$.

    \smallskip\noindent
    {\em Case~2.} There exists an integer~$k\geq 0$ s.t.\ $\sigma_y^k(p)\mid u$. Then for every integer~$\ell$, we have $\gcd(\sigma_y^\ell(p),v)=1$, since~$K$ is shift-reduced w.r.t.\ $y$. Let
    \[m = \max\{i\in \bN\mid \sigma_y^i(p)\mid u\} +1.\]
    One can see that~$\sigma_y^m(p)$ is strongly prime with~$K$.

    \smallskip\noindent
    {\em Case~3.} There exists an integer~$k\leq 0$ s.t.\ $\sigma_y^k(p)\mid v$. Then for every integer~$\ell$, we have $\gcd(\sigma_y^\ell(p),u)=1$, since~$K$ is shift-reduced w.r.t.\ $y$. Let
    \[m = \min\{i\in \bN\mid \sigma_y^i(p)\mid v\}-1.\]
    One can see that~$\sigma_y^m(p)$ is strongly prime with~$K$.

    The lemma follows.
\end{proof}
The following proposition computes a common multiple of some sequence~$\{b_i\}_{i\geq 0}$ satisfying~\eqref{EQ:ithdecom}, provided that~$b_0$ is integer-linear.
\begin{prop}\label{PRO:commonden}
With Convention~\ref{CON:convention}, assume that
\begin{equation}\label{EQ:initialred}
    T = \Delta_y(gH) + \left(\frac{a}{b} + \frac{q}{v} \right) H,
\end{equation}
where~$g\in \cC(x,y)$, $a,b\in \cC(x)[y]$, $\deg_y(a) < \deg_y(b)$, $\gcd(a,b)=1$, $b$ is shift-free w.r.t.\ $y$ and strongly prime with~$K$, and~$q\in \bW_K$. Further assume $b$ is integer-linear. Then there exists a polynomial~$B\in \cC(x)[y]$ with~$B$ shift-free w.r.t.\ $y$ and strongly prime with~$K$ s.t.\ $b\mid B$, and for every nonnegative integer $i\geq 0$,~$\sigma_x^i(T)$ can be decomposed into
\begin{equation}\label{EQ:ithred}
    \sigma_x^i(T) = \Delta_y(g_iH) + \left(\frac{a_i}{B} + \frac{q_i}{v}\right)H,
\end{equation}
where~$g_i\in \cC(x,y)$,~$a_i \in \cC(x)[y]$ with $\deg_y(a_i) < \deg_y(B)$ and~$q_i \in \bW_K$.
\end{prop}
\begin{proof}
   When~$b\in \cC(x)$. By the modified Abramov-Petkov{\v s}ek reduction, \eqref{EQ:ithdecom} holds for every~$i>0$. Then~$b_i\in \cC(x)$ by Proposition~\ref{PRO:brelationship}. The proposition follows by letting~$B=1$. Assume that~$b\notin \cC(x)$. Since~$b$ is integer-linear, by~\eqref{EQ:integerlineardecomp},
\begin{equation}\label{EQ:b}
 b=c_b \cdot h_{1}^{\xi_{1}}\cdots h_{t}^{\xi_{t}}
\end{equation}
where $c_b \in \cC(x)$, $t\in \bN$, each $h_j$ is a monic irreducible integer-linear polynomial of the form $h_j = P_j(\lambda_j x + \mu_j y)$ for~$P_j \in \cC[z]$, $\lambda_j, \mu_j\in \bZ$, $\mu_j> 0$ and $\gcd(\lambda_j, \mu_j) = 1$, and~$\xi_{j}$ belongs to $\bN[\delta^{(\lambda_j, \mu_j)}]$ for $1\leq j \leq t$. Moreover, $h_j \nsim_{x,y} h_k $ whenever $j \neq k$. Due to the primeness of~$h_j$ and the partial fraction decomposition of~$a/b$, it suffices to prove the local case, that is, $c_b = t=1$ and
\[b = P(\lambda x + \mu y)^{\sum_{i=0}^s c_i \delta^i},\]
where~$P\in \cC[z]$, $s, c_i\in \bN$, $\lambda, \mu \in \bZ$, $\mu >0$ and $\delta = \delta^{(\lambda, \mu)}$. Note that for every~$i\in \bN$, there are unique integers~$j, k_j$ with $0\leq j \leq \mu-1$ s.t.\ $i = \mu k_j + j$. Let~$c_j'= c_{\mu k_j + j}$. Since~$b$ is shift-free w.r.t.~$y$, we have
\[b = \prod_{j=0}^{\mu-1}P(\lambda x + \mu y+j)^{c_j'\sigma_y^{k_j}}.\]
For every~$0\leq j \leq \mu-1$, set~$\ell_j$ to be $k_j$ if~$m_j \neq 0$, or some integer otherwise with~$P(\lambda x+ \mu y+ j)^{\sigma_y^{\ell_j}}$ strongly prime with~$K$ by Lemma~\ref{LEM:stronglyprime}. Let $m = \max_{0\leq j \leq \mu-1}\{c_j'\}$ and
\begin{equation}\label{EQ:commonden}
    B = \prod_{j=0}^{\mu-1}P(\lambda x + \mu y+j)^{m\sigma_y^{\ell_j}}.
\end{equation}
Since~$\ell_j = k_j$ when~$m_j\neq 0$, every irreducible factor of~$b$ divides~$B$ and thus~$b\mid B$ by the maximum of~$m$. $B$ is shift-free w.r.t.\ $y$ since~$0\leq j\leq \mu-1$. Moreover, $B$ is strongly prime with~$K$ by the choice of~$\ell_j$.

It remains to show that~\eqref{EQ:ithred} holds for every nonnegative integer~$i$. To prove this, we first show $\sigma_x(B)\approx_y B$. By~\eqref{EQ:commonden},
\[B \approx_y \prod_{j=0}^{\mu-1}P(\lambda x + \mu y+j)^m,\]
which establishes that
\begin{equation*}
 \sigma_x(B) \approx_y\prod_{j=0}^{\mu-1}P(\lambda x + \mu y+j+\lambda)^m.
\end{equation*}
One sees that there is a unique integer~$0\leq k \leq \mu-1$ s.t.
\[P(\lambda x + \mu y+j+\lambda) \sim_y P(\lambda x + \mu y+k).\]
Conversely, for any~$0\leq k\leq \mu-1$, there exists a unique integer~$0\leq j\leq \mu-1$ s.t.\ the above equivalence holds. Thus
\[\sigma_x(B)\approx_y \prod_{k=0}^{\mu-1}P(\lambda x + \mu y+k)^m\approx_y B.\]

For~$i=0$, letting~$a_0 = aB/b$ and~$q_0 = q$ gives~\eqref{EQ:ithred}.
For every~$i > 0$, $\sigma_x^i(B)\approx_y \sigma_x^{i-1}(B)$ since~$\sigma_x(B)\approx_y B$ and then~$\sigma_x^i(B)\approx_y B$. By the modified Abramov-Petkov{\v s}ek reduction, \eqref{EQ:ithdecom} holds for every~$i\geq 0$, in which~$b_0 = b$. According to Proposition~\ref{PRO:brelationship}, $b_i\approx_y\sigma_x^i(b_0)$. It follows from~$b\mid B$ that $\sigma_x^i(b)\mid \sigma_x^i(B)$.  Consequently, we have
\[b_i\approx_y \sigma_x^i(b)\mid \sigma_x^i(B)\approx_y B.\]
Thus there is $\tb_{i}\in \cC(x)[y]$ dividing~$B$ so that~$\tb_{i} \approx_y b_{i}$. Moreover, $\tb_i$ is strongly prime with~$K$ as~$B$ is. It follows from Theorem~5.6 in \cite{CHKL2015} that there exist $\tilde{g_i}\in \cC(x,y)$, $\ta_{i}\in \cC(x)[y]$ with $\deg_y(\ta_{i}) < \deg_y(\tb_{i})$ and~$\tq_{i} \in \bW_K$ such that~$\sigma_x^i(T) = \Delta_y(\tilde{g_i}H)+(\ta_i/\tb_i+\tq_i/v)H$. The assertion follows by noticing
\[\sigma_x^i(T) = \Delta_y(\tilde{g_i}H) + \left(\frac{\ta_i B/\tb_i}{B} + \frac{\tq_i}{v}\right)H.\]
\end{proof}
Under the assumptions and notations of Proposition~\ref{PRO:commonden}, applying the modified Abramov-Petkov{\v s}ek reduction to~$T$ w.r.t.\ $y$ yields~\eqref{EQ:map}.
By Proposition~\ref{PRO:uniqueness} and Proposition~\ref{PRO:sf}, the significant denominator of~$r$ in~\eqref{EQ:map} is shift-related to~$b$ w.r.t.\ $y$. Thus the shift-equivalence classes represented by~$h_j\, (1\leq j \leq t)$ in~\eqref{EQ:b} are independent of the choice of~$b$. Therefore, the degree of $B$ w.r.t.\ $y$ is fixed once a hypergeometric term $T$ is given, although the form of~$B$ depends on the choice of~$b$.
\subsection{Upper and lower bounds}\label{SUBSEC:bounds}
Now we show that Proposition~\ref{PRO:commonden} implies some residual forms $\{a_i/b_i+q_i/v\}_{i\geq 0}$ satisfying~\eqref{EQ:ithdecom} form a finite-dimensional vector space over~$\cC(x)$, and then derive the order bounds for minimal telescopers.
\begin{theorem}\label{THM:upperbound}
	With the assumptions and notations introduced in Proposition~\ref{PRO:commonden}, the order of a minimal telescoper for $T$ is no more than
	\begin{align*}\label{EQ:upperbound}
		& \max\{\deg_y(u), \deg_y(v)\} - \llbracket \deg_y(v-u) \leq \deg_y(u) -1\rrbracket\nonumber\\
		&+\sum_{j=1}^t{\mu_j m_j \deg(P_j)},
	\end{align*}
   where $m_j$ is the maximum coefficient of $\xi_{j}$ for $1\leq j \leq t$, and $\llbracket\varphi\rrbracket$ equals~$1$ if $\varphi$ is true, otherwise it is~$0$.
\end{theorem}
\begin{proof}
	Let $L=\sum_{i=0}^{\rho}{e_i S_x^i}$ with~$\rho \in \bN$, $e_0, \ldots ,e_\rho \in \cC(x)$, not all zero, be a minimal telescoper for $T$ w.r.t.\ $y$. By Proposition~\ref{PRO:commonden}, \eqref{EQ:ithred} holds for every $0\leq i \leq \rho$. Then by the arguments in Section~\ref{SEC:telescoping}, the linear system \eqref{EQ:system}, in which $b_i = B_\rho = B$ for $1\leq i \leq \rho$, of equations for the variables~$\{e_0, \ldots, e_\rho\}$ has a nontrivial solution in~$\cC(x)^{\rho+1}$. Since~$Q_\rho\in \bW_K$, the number of terms w.r.t.~$y$ in $Q_\rho$ is no more than $\dim_{\cC(x)}(\bW_K)$, which is bounded by
	\[\max\{\deg_y(u), \deg_y(v)\} - \llbracket \deg_y(v-u) \leq \deg_y(u) -1\rrbracket\]
according to Proposition~4.7 in \cite{CHKL2015}. Note that The solutions of the system \eqref{EQ:system} are in one-to-one correspondence with the telescopers for $T$. Comparing coefficients of like powers of $y$ of the linear system \eqref{EQ:system} yields at most $\deg_y(A_\rho)+\dim_{\cC(x)}{\bW_K}+1$ equations. Hence this system has nontrivial solutions whenever $\rho>\deg_y(A_\rho)+\dim_{\cC(x)}{\bW_K}$. It implies that the order of a minimal telescoper for $T$ is no more than~$\deg_y(A_\rho)+\dim_{\cC(x)}{\bW_K}+ 1$. The theorem follows by applying~$\deg_y(A_\rho) < \deg_y(B) = \sum_{j=1}^t{\mu_j m_j \deg_y(P_j)}$.
\end{proof}
In addition, we can further obtain a lower bound for the order of telescopers for~$T$.
\begin{theorem}\label{THM:lowerbound}
With the assumptions of Proposition~\ref{PRO:commonden}, further assume that $T$ is not summable w.r.t.\ $y$. Then the order of a telescoper for $T$ is at least
 \begin{equation*}\label{EQ:lowerbound}
   \max_{\begin{array}{c}
	  \scriptscriptstyle p \mid b, \text{ multi.} \alpha\\[-1ex]
	  \scriptscriptstyle \text{ irred.\ {\em \&} monic}\\[-1ex]
	  \scriptscriptstyle \deg_y(p) > 0
         \end{array}}\hspace{-5pt}
 \min \left\{\rho\in \bN\setminus\{0\}: \begin{array}{c}
                                            \sigma_y^\ell (p)^\alpha \mid \sigma_x^\rho(b) \\[1ex]
                                            \text{ for some } \ell \in \bZ
                                           \end{array}
      \right\}.
 \end{equation*}
\end{theorem}
\begin{proof}
Let $L=\sum_{i=0}^{\rho}{e_i S_x^i}$ with~$\rho\in \bN$, $e_0, \ldots ,e_\rho \in \cC(x)$, not all zero, be a minimal telescoper for $T$. Then $\rho \geq 1$ as~$T$ is not summable. By the modified Abramov-Petkov{\v s}ek reduction, we have \eqref{EQ:ithdecom} holds for~$0\leq i \leq \rho$, in which~$b_0=b$. Note that $L$ is a minimal telescoper, so~$e_0 \neq 0 $ and
\[e_0\frac{a_0}{b_0} \ +\ e_1\frac{a_1}{b_1}\ +\ \cdots\ +\ e_\rho \frac{a_\rho}{b_\rho}=0\]
by the system~\eqref{EQ:system}. By partial fraction decomposition, for any monic irreducible factor $p$ of $b_0$ with $\deg_y(p) > 0$ and multiplicity~$\alpha>0$, there exists an integer $i$ with $1\leq i \leq \rho$ so that $p^\alpha$ is also a factor of~$b_i$. According to Proposition~\ref{PRO:brelationship}, $b_i\approx_y \sigma_x^i(b_0)$. Thus there is a factor $p'$ of~$\sigma_x^{i}(b_0)$ with multiplicity at least~$\alpha$ s.t.\ $p' \sim_y p$.  Let $i_{p}$ be the minimal one with this property. Then the assertion follows by the fact that for each factor~$p$ of $b_0=b$, there exist no telescopers for~$T$ of order less than~$i_p$.
\end{proof}
With the upper and lower bounds, one may try to analyze the complexity of the algorithm {\sf ReductionCT} from \cite{CHKL2015} and further improve it by adding them into the procedures. Instead of doing so, in the rest of this paper, we are going to compare our bounds to the known ones in the literature.

\section{Comparison of bounds}\label{SEC:comparsion}
Upper and lower bounds for the order of telescopers for hypergeometric terms have been studied in \cite{MoZe2005} and \cite{AbLe2005}, respectively. In this section, we are going to review these known bounds and then compare them to our bounds.
\subsection{Apagodu-Zeilberger upper bound}\label{SUBSEC:AZbound}
Let $T$ be a {\em proper} hypergeometric term over $\cC(x,y)$, that is it can be written in the form
\begin{equation}\label{EQ:properht}
T = p w^x z^y\prod_{i= 1}^m\frac{(\alpha_i x + \alpha_i' y + \alpha_i''-1)!(\beta_i x -\beta_i' y + \beta_i''-1)!}{(\mu_i x + \mu_i' y + \mu_i''-1)! (\nu_i x-\nu_i' y+ \nu_i''-1)!},
\end{equation}
where $p\in \cC[x,y]$, $m\in \bN$ is fixed, $\alpha_i, \alpha_i', \beta_i, \beta_i', \mu_i, \mu_i', \nu_i, \nu_i'$ are non-negative integers and $w, z, \alpha_i'', \beta_i'', \mu_i'', \nu_i''\in \cC$ . Further assume that there exist no integers~$i,j$ with $1\leq i, j \leq m $ s.t.
\begin{align*}
\begin{array}{cccccc}
          & \{ \alpha_i = \mu_j & \& & \alpha_i' = \mu_j' & \& & \alpha_i''-\mu_j'' \in \bN \}\\[1ex]
 \text{or}& \{ \beta_i = \nu_j & \& & \beta_i' = \nu_j' & \& & \beta_i''-\nu_j'' \in \bN \}.
\end{array}
\end{align*}
We refer this as the \emph{generic} situtation.
Then Apagodu and Zeiberger \cite{MoZe2005} stated that the order of a minimal telescoper for $T$ is bounded by
\[B_{AZ} = \max\left\{\sum_{i=1}^m(\alpha_i'+\nu_i'), \sum_{i=1}^m(\beta'_i+\mu_i')\right\}.\]
We now show that $B_{AZ}$ given above is at least the order bound on minimal telescopers for $T$ obtained from Theorem~\ref{THM:upperbound}. Reordering the factorial terms in~\eqref{EQ:properht} if necessary, let $\cS$ be the maximal set of integers $i$ with~$1\leq i\leq m$ satisfying
\begin{align*}
\begin{array}{cccccc}
          & \{ \alpha_i = \mu_i & \& & \alpha_i' = \mu_i' & \& & \mu_i''-\alpha_i'' \in \bN \}\\[1ex]
 \text{or}& \{ \beta_i = \nu_i & \& & \beta_i' = \nu_i' & \& & \nu_i''-\beta_i'' \in \bN \}.
\end{array}
\end{align*}
Rewrite $T$ as
\[r w^x z^y
      \prod_{\begin{array}{c}
              \scriptscriptstyle i = 1\\[-1ex]
              \scriptscriptstyle i \notin \cS
             \end{array}
             }^m\frac{(\alpha_i x + \alpha_i' y + \alpha_i''-1)!(\beta_i x -\beta_i' y + \beta_i''-1)!}{(\mu_i x + \mu_i' y + \mu_i''-1)! (\nu_i x-\nu_i' y+ \nu_i''-1)!},\]
where $r \in \cC(x,y)$. For $q \in \cC[x,y]$, and $m \in \bN$, let
\[q^{\overline{m}} = q (q+1) (q+2)\cdots (q+m-1)\]
with the convention $q^{\overline{0}} = 1$. It can be calculated that the kernel and shell of the shift quotient $\sigma_y(T)/T$ are
\begin{align*}
 K = z \prod_{i
             } \frac{(\alpha_i x + \alpha_i' y + \alpha_i'')^{\overline{\alpha'_i}} (\nu_i x-\nu_i' y+ \nu_i''-\mu_i')^{\overline{\nu'_i}}}{(\mu_i x + \mu_i' y + \mu_i'')^{\overline{\mu'_i}}(\beta_i x -\beta_i' y + \beta_i''-\beta_i')^{\overline{\beta'_i}}},
\end{align*}
where the product runs over all $i$ from $1$ to $m$ s.t.\ $ i \notin \cS$, $\alpha_i',\beta_i'>0$ and $\mu_i', \nu_i' > 0$, and $S=r$, respectively. Let~$K = u/v$ with $\gcd(u,v)=1$. Note that $K$ is proper, a straightforward calculation implies
\[\deg_y(u) = \sum_{i=1,i\notin\cS}^m (\alpha_i'+\nu_i') \text{ and } \deg_y(v) = \sum_{i = 1, i\notin \cS}^m(\beta_i'+ \mu_i').\]
Applying the modified Abramov-Petkov{\v s}ek reduction to~$T$ w.r.t.\ $y$ yields~\eqref{EQ:initialred}, in which~$b$ is integer-linear. Since $b$ comes from the shift-free part of the denominator of~$r$, it factors into integer-linear polynomials of degree~$1$ which are shift-equivalent to either~$(\mu_i x+ \mu_i' y+\mu_i'')$ or $(\beta_i x -\beta_i' y+\beta_i'')$ w.r.t.\ $x,y$ for some~$i \in \cS$. Note that each~$i$ in~$\cS$ increases the multiplicities of the corresponding integer-linear factors in $b$ by at most 1.
Hence, the bound given in Theorem~\ref{THM:upperbound} is no more than
\begin{align*}
&\quad \max\{\deg_y(u), \deg_y(v)\} - \llbracket \deg_y(v-u) \leq \deg_y(u) -1\rrbracket\\
&+ \sum_{i=1,i\in\cS}^m (\beta_i'+\mu_i'),
\end{align*}
which is exactly equal to $B_{AZ} - \llbracket \deg_y(v-u) \leq \deg_y(u) -1\rrbracket$ since $\sum_{i=1,i\in\cS}^m (\alpha_i'+\nu_i')=\sum_{i=1,i\in\cS}^m (\beta_i'+\mu_i')$.

In general, i.e., when we have the generic situation, the order bound in Theorem \ref{THM:upperbound} is almost the same as~$B_{AZ}$. However, our bound may be much tighter in some special examples.
\begin{example}\label{EX:sharperub}
 Consider a rational function
 \[T=\frac{\alpha^2 y^2+\alpha^2 y-\alpha \beta y+2 \alpha x y+x^2}{(x+\alpha y+\alpha)(x+\alpha y) (x+\beta y)},\]
 where $\alpha, \beta$ are positive integers and $\alpha \neq \beta$. Rewriting $T$ into the proper form yields $B_{AZ} = \alpha + \beta$. On the other hand, the kernel of~$\sigma_y(T)/T$ is $1$ since $T$ is a rational function. By the modified Abramov-Petkov{\v s}ek reduction, $b=x+\beta y$ in~\eqref{EQ:initialred}. According to Theorem~\ref{THM:upperbound}, a minimal telescoper for~$T$ has order no more than $\beta$, which is indeed the real order of minimal telescopers for $T$.
\end{example}
\begin{remark}\label{REM:nonproper}
  Only with~\cite[Theorem 10]{Abra2003}, the upper order bound on minimal telescopers derived in~\cite{MoZe2005} can be also applied to non-proper hypergeometric terms. On the other hand, Theorem~\ref{THM:upperbound} can be applied to any hypergeometric term provided that its telescopers exist.
\end{remark}

\subsection{Abramov-Le lower bound}\label{SUBSEC:ALbound}
With Convention \ref{CON:convention}, assume that $T$ has the initial reduction~\eqref{EQ:initialred}, in which~$b$ is integer-linear. Define $H' {=} H/v$. A direct calculation leads to~$\sigma_y(H')/ H' = u/\sigma_y(v)$,
which is again shift-reduced w.r.t.\ $y$. Let $d'\in \cC(x)[y]$ be the denominator of $\sigma_x(H')/H'$.
Then the algorithm {\sf LowerBound} in \cite{AbLe2005} asserts that the order of telescopers for $T$ is at least
\begin{align*}
B_{AL} =  \hspace{-8pt}\max_{{\begin{array}{c}
		    \scriptscriptstyle p \mid b\\[-1ex]
		    \scriptscriptstyle \text{ irred.\ \& monic}\\[-1ex]
		    \scriptscriptstyle \deg_y(p) > 0
		  \end{array}
		 }}\hspace{-8pt}
	    \min \left\{\rho\in \bN\setminus\{0\}:
	    \begin{array}{c}
	      \sigma_y^\ell(p) \mid \sigma_x^{\rho} (b) \\
	      \text{ or }\\
	      \sigma_y^\ell(p) \mid \sigma_x^{\rho-1}(d')\\[1ex]
	      \text{ for some } \ell \in \bZ
	    \end{array}\right\}
\end{align*}
Comparing to $B_{AL}$, it is obvious that the lower bound given by Theorem~\ref{THM:lowerbound} can be better but not worse than~$B_{AL}$.
\begin{example}\label{EX:sharperlb}
 Consider a hypergeometric term
 \[T = \frac{1}{(x-\alpha y-\alpha)(x-\alpha y-2)!},\]
 where $\alpha\in \bZ$ and $\alpha \geq 2$. By the algorithm {\sf LowerBound}, a telescoper for~$T$ has order at least $2$. On the other hand, a telescoper for $T$ has order at least~$\alpha$ by Theorem~\ref{THM:lowerbound}. In fact,~$\alpha$ is exactly the order of minimal telescopers for $T$.
\end{example}
\section{Acknowledgement}
I am very grateful to my advisors Manuel Kauers and Ziming Li for their helpful discussions and valuable comments. I also would like to thank Shaoshi Chen for his suggestions and support. Besides, I wish to express my gratitude to anonymous referees for many useful suggestions.

%
\end{document}